\documentclass[sigconf]{acmart}
\usepackage{xfrac}
\usepackage{graphicx}
\usepackage{tabularx}
\usepackage{subcaption}
\usepackage{amssymb,amsthm}
\usepackage{stackengine}
\usepackage{url}

\def\pproblem#1#2#3#4{
\begin{flushleft} 
  \noindent 
  {\sc #1}\\
  {\bf Instance: }#2.\\
  {\bf Promise: }#3.\\
  {\bf Question: }#4?\\
\end{flushleft}}

\newcommand{\SNNHalting}{\ensuremath{\mathcal{S}(R_T,R_S)}-{\sc Halting}}
\newcommand{\SNNAccept}{\ensuremath{\mathcal{S}((\mathcal{O}(n),\mathcal{O}(n)),R_S)}-{\sc Net\-work Halting}}
\newcommand{\NwFlow}{{\sc Network Flow}}

\AtBeginDocument{%
  \providecommand\BibTeX{{%
    \normalfont B\kern-0.5em{\scshape i\kern-0.25em b}\kern-0.8em\TeX}}}

\begin{document}

\title[Complexity of SNNs]{On the computational power and complexity of Spiking Neural Networks}

\author{Johan Kwisthout}
\authornote{Supported by a grant from Intel Corporation.}
\email{j.kwisthout@donders.ru.nl}
\author{Nils Donselaar}
\authornote{Supported by NWO grant 612.001.601.}
\email{n.donselaar@donders.ru.nl}
\affiliation{%
  \institution{Donders Institute for Brain, Cognition, and Behaviour}
  \streetaddress{Montessorilaan 3}
  \postcode{6525 HR}
  \city{Nijmegen}
}

\begin{abstract}
The last decade has seen the rise of neuromorphic architectures based on artificial spiking neural networks, such as the SpiNNaker, TrueNorth, and Loihi systems. The massive parallelism and co-locating of computation and memory in these architectures potentially allows for an energy usage that is orders of magnitude lower compared to traditional Von Neumann architectures. However, to date a comparison with more traditional computational architectures (particularly with respect to energy usage) is hampered by the lack of a formal machine model and a computational complexity theory for neuromorphic computation. In this paper we take the first steps towards such a theory. We introduce spiking neural networks as a machine model where---in contrast to the familiar Turing machine---information and the manipulation thereof are co-located in the machine. We introduce canonical problems, define hierarchies of complexity classes and provide some first completeness results. 
\end{abstract}



\maketitle

\section{Introduction}

Moore's law \cite{Moore75} stipulates that the number of transistors in integrated circuits (ICs) doubles about every two years. With transistors becoming faster IC performance doubles every 18 months, at the cost of increased energy consumption as transistors are added \cite{Shalf15}. Moore's law is slowing down and is expected\footnote{\url{https://www.economist.com/technology-quarterly/2016-03-12/after-moores-law}.} to end by 2025. Traditional (``Von Neumann'') computer architectures separate computation and memory by a bus, requiring both data and algorithm to be transferred from memory to the CPU with every instruction cycle. This has been described, already in 1978, as the Von Neumann bottleneck \cite{Backus78}. While CPUs have grown faster, transfer speed and memory access lagged behind \cite{Hennessy11}, making this bottleneck an increasingly difficult obstacle to overcome.

In summary, while more data than ever before is produced, we are simultaneously faced with the end of Moore's law, limited performance due to the Von Neumann bottleneck, and an increasing energy consumption (with corresponding carbon footprint) \cite{OakRidge16}. These issues have accelerated the development of several generations of so-called neuromorphic hardware \cite{Mead90,Indivieri11,Davies18}. Inspired by the structure of the brain (largely parallel computations in neurons, low power consumption, event-driven communication via synapses) these architectures co-locate computation and memory in artificial (spiking) neural networks. The spiking behavior allows for potentially energy-lean computations \cite{Maass14} while still allowing for in principle any conceivable computation \cite{Maass96}. However, we do not yet fully understand the potential (and limitations) of these new architectures. Benchmarking results are suggesting that event-driven information processing (e.g. in neuromorphic robotics or brain-computer-interfacing) and energy-critical applications might be suitable candidate problems, whereas `deep' classification and pattern recognition (where spiking neural networks are outperformed by convolutional deep neural networks) and applications that value precision over energy usage may be less natural problems to solve on neuromorphic hardware. Although several algorithms have been developed to tackle specific problems, there is currently no insight in the potential and limitations of neuromorphic architectures. 

The emphasis on {\em energy} as a vital resource, in addition to the more traditional {\em time} and {\em space}, suggests that the traditional models of computation (i.e., Turing machines and Boolean circuits) and the corresponding formal machinery (reductions, hardness proofs, complete problems etc.) are ill-matched to capture the computational power of spiking neural networks. What is lacking is a unifying computational framework and structural complexity results that can demonstrate what can and cannot be done in principle with bounded resources with respect to convergence time, network size, and energy consumption \cite{Haxhimusa14}. Previous work is mostly restricted to variations of Turing machine models within the Von Neumann architecture \cite{Graves14} or energy functions defined on threshold circuits \cite{Uchizawa09} and as such unsuited for studying spiking neural networks. This is nicely illustrated by the following quote:

\begin{quote}``It is \ldots likely that an entirely new computational theory paradigm will need to be defined in order to encompass the computational abilities of neuromorphic systems'' \cite[p.29]{OakRidge16}
\end{quote}

In this paper we propose a model of computation for spiking neural network-based neuromorphic architectures and lay the foundations for a {\em neuromorphic complexity theory}. In Section \ref{Machine_model} we will introduce our machine model in detail. In Section \ref{Resources} we further elaborate on the resources time, space, and energy relative to our machine model. In Section \ref{Complexity_classes} we will explore the complexity classes associated with this machine model and derive some basic structural properties and hardness results. We conclude the paper in Section \ref{Conclusion}.

\section{Machine model}
\label{Machine_model}

In order to abstract away from the actual computation on a neuromorphic device, in a similar vein as the Turing machine acts as an abstraction of computations on traditional hardware architectures, we introduce a novel notion of computation based on spiking neural networks. We will first elaborate on the network model and then proceed to translate that to a formal machine model.

\subsection{Spiking neural network model}

We will first introduce the specifics of our spiking neural network model, which is a variant of the leaky integrate-and-fire model introduced by Severa and colleagues at Sandia National Labs \cite{Severa16}. This model defines a discrete-timed spiking neural network as a labeled finite digraph $\mathcal{S} = (N, S)$ comprised of a set of neurons $N$ as vertices and a set of synapses $S$ as arcs. Every neuron $k\in N$ is a triple $(T_k \in \mathbb{Q}_{\geq0}, R_k \in \mathbb{Q}_{\geq0}, m_k \in [0,1])$ representing respectively threshold, reset voltage, and leakage constant, while a synapse $s\in S$ is a $4$-tuple $(k \in N, l \in N, d \in \mathbb{N}_{>0}, w \in \mathbb{Q})$ for the pre-synaptic neuron, post-synaptic neuron, synaptic delay and weight respectively. We will use notation $S_{k,l} = (d,w)$ as a shorthand to refer to specific synapses and shorthands $d_{kl}$ and $w_{kl}$ to refer to the synaptic delay and weight of a specific synapse.

The basic picture is thus that any spikes of a neuron $k$ are carried along outgoing synapses $S_{k,l}$ to serve as inputs to the receiving neurons $l$. The behavior of a spiking neuron $k$ at time $t$ is typically defined using its membrane potential $u_k (t) = m_ku_k(t-1) + \sum_j w_{jk}x_j(t - d_{jk}) + b_k$ which is the integrated weighted sum of the neuron's inputs (taking into account synaptic delay) plus an additional bias term. Whether a neuron spikes or not at any given time is dependent on this membrane potential, either deterministically (i.e., the membrane potential acts as a threshold function for the spike) or stochastically (i.e., the probability of a spike being released is proportional to the potential); in this paper we assume deterministic spike responses. A spike $x_k(t)$ is abstracted here to be a singular discrete event, that is, $x_k(t) = 1$ if a spike is released by neuron $x_k$ at time $t$ and $x_k(t) = 0$ otherwise.
Figure \ref{snn_model} gives an overview of this spiking neuron model.

One can also define the spiking behavior of a neuron {\em programmatically} rather than through its membrane potential, involving so-called spike trains, i.e.~predetermined spiking schedules. Importantly, such neurons allow for a means of providing the input to a spiking neural network. Furthermore, for regular (non-programmed) neurons the bias term can be replaced by an appropriately weighted connection stemming from a continuously firing programmed neuron; for convenience this bias term will thus be omitted from the model. Figure \ref{snn_legend} introduces our notational conventions that we use for graphically depicting networks, along with a few simple networks as an illustration. As a convention, unless otherwise depicted, neuron and synapse parameters have their default values $R = 0$ and $T = m = d = w = 1$. 
%
%
%
%

\begin{figure}[t]
\centering
\includegraphics[width=7.5cm]{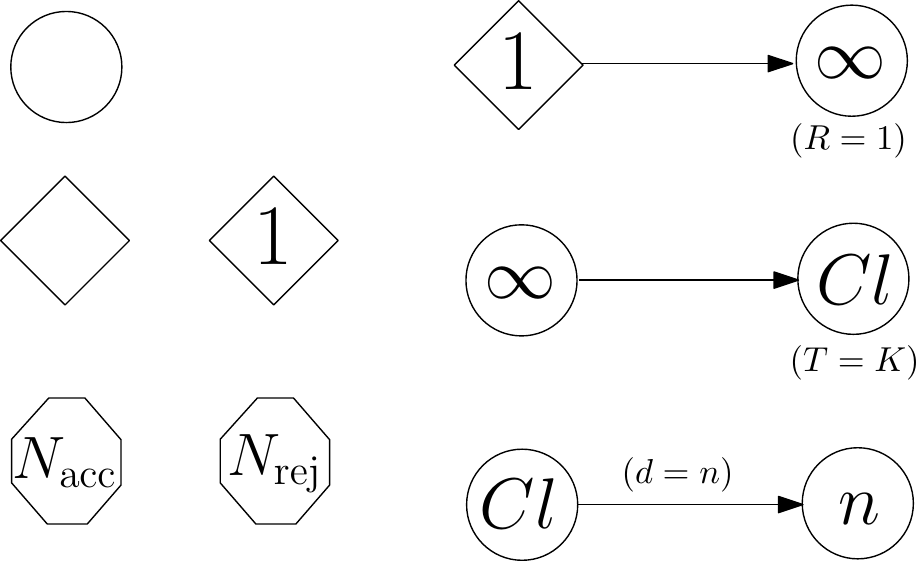}
\caption{Notational conventions for (top to bottom on the left) a regular neuron, a programmed neuron, dedicated notation for programmed neurons firing once at timestep $t=0$, and dedicated acceptance and rejection neurons. To the right we show simple circuits realizing a continuously firing neuron, a clock neuron firing every $K$ time steps, and a temporal representation of a natural number $n < K$ relative to a clock.}
\label{snn_legend}
\rule{\columnwidth}{0.3mm}
\vspace{1mm}
\end{figure}

\begin{figure*}[t]
\centering
\includegraphics[width=16cm]{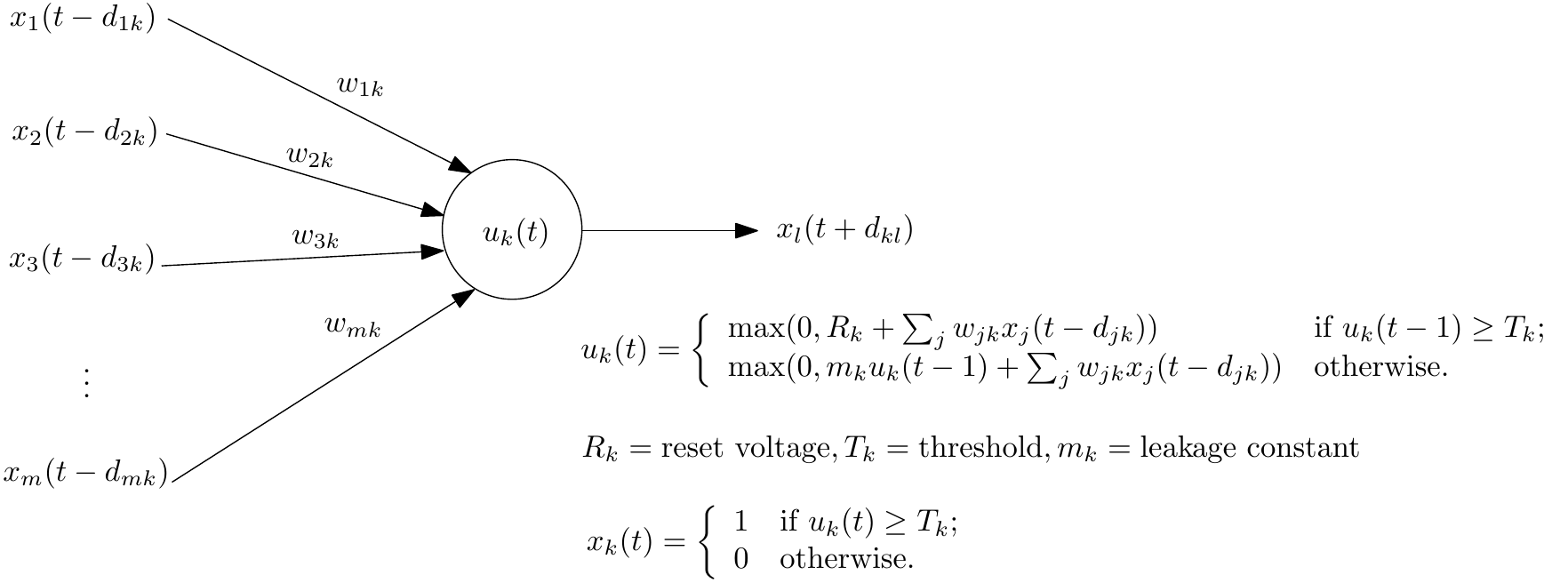}
\caption{A spiking neuron model with deterministic spiking behavior, describing the membrane potential $u(t)$ of a leaky integrate-and-fire neuron $k$ over time, based on the integrated weighted sum of incoming post-synaptic potentials. We enforce that the membrane potential is non-negative. Spikes are emitted when the membrane potential reaches its threshold and arrive at post-synaptic neurons $l$ with synaptic delay $d_{kl}$.}
\label{snn_model}
\rule{\columnwidth}{0.3mm}
\vspace{1mm}
\end{figure*}

For every spiking neural network $\mathcal{S}$ we require the designation of two specific neurons as the acceptance neuron $N_{\mathrm{acc}}$ and the rejection neuron $N_{\mathrm{rej}}$. The idea is that the firing of the corresponding neuron signifies acceptance and rejection respectively, at which point the network is brought to a halt. In the absence of either one of those neurons, we can impose a time constraint and include a new neuron which fires precisely when $N_{\mathrm{acc}}$ or $N_{\mathrm{rej}}$ (whichever is present) did not fire within time, thus adding the missing counterpart. In this way, we ensure that this model is a specific instantiation of Wolfgang Maass' generic spiking neural network model that was shown to be Turing complete \cite{Maass96}; hence, these spiking neural networks can in principle (when provided the necessary resources) compute anything a Turing machine can. More interesting is the question whether we can design smart algorithms that minimize the use of resources, for example, minimize energy usage within given bounds on time and network size. In order to answer this question we need to define a suitable formal abstraction of what constitutes a computational problem on a spiking neural network.

\subsection{Canonical problems}
\label{Canonical_problems}

Canonical computational problems on Turing machines typically take the following form: ``Given machine $\mathcal{M}$ and input $i$ on its tape, does $\mathcal{M}$ accept $i$ using resources at most $R$''? Here, $L$ is the language that $\mathcal{M}$ should accept, and the job of $\mathcal{M}$ is to decide whether $i \in L$. To translate such problems to a spiking neural network model one needs to define the machine model $\mathcal{S}$, the resources $R$ that $\mathcal{S}$ may use, how the input $i$ is encoded and what it means for $\mathcal{S}$ to accept the input $i$ using resources $R$. 

This is a non-trivial problem. In a Turing machine the input is typically taken to be encoded in binary notation and written on the machine's tape, while the algorithm for accepting inputs $i$ is represented by the state machine of $\mathcal{M}$. However, in spiking neural networks both the problem input and the algorithm operating on it are encoded in the network structure and parameters. While the most straightforward way of encoding the input is through programming a spike train on a set of input neurons, in some cases it might be more efficient to encode it otherwise, such as at the level of synaptic weights or even delays. In that sense a spiking neural network is different from both a Turing machine and a family of Boolean circuits as depicted in Table \ref{overview_machines}.

\begin{table*}[t]
\centering
\begin{tabularx}{\linewidth}{|p{44pt}|p{60pt}|X|X|X|}
\hline
 & Character of device(s) & Input representation $i$ & Resources $R$ & Canonical problem $Q$ \\
\hline
 Turing Machine $\mathcal{M}$ & One machine deciding all instances $i$. & Input is presented on the machine's tape. & Time, size of the tape, transition properties, acceptance criteria. & Does $\mathcal{M}$ decide whether $i \in L$ using resources at most $R$? \\
\hline
 Family of Boolean circuits $\mathcal{C}_{|i|}$ & One circuit for every input size $|i|$. & Input is represented as special gates. & Circuit size and depth, size and fan-in of the gates. & Does, for each $i$, the corresponding circuit $\mathcal{C}_{|i|}$ decide whether $i \in L$ using resources at most $R$? \\
\hline
 Collection of SNNs $\mathcal{S}_i$ & One network for every input $i$ or set of inputs $\{i_1, \ldots, i_m\}$. & Input is encoded in the network structure and/or presented as spike trains on input neurons. & Network size, time to convergence, total number of spikes. & Is there a resource-bounded Turing Machine $\mathcal{M}$ that, given $i$, generates (using resources $R_T$) $\mathcal{S}_i$ which decides whether $i \in L$ using resources at most $R_S$? \\
\hline
\end{tabularx}
\caption{Overview of machine models: Turing machines, Boolean circuits, and families of spiking neural networks.}
\label{overview_machines}
\end{table*}

Hence, we introduce a novel computational abstraction, suitable for describing the behavior of neuromorphic architectures based on spiking neural networks. We postulate that a network $\mathcal{S}_i$ encodes both the input $i$ and the algorithm deciding whether $i \in L$. What it means to decide a problem $L$ using a spiking neural network now becomes the following: that there is an $R_T$-resource-bounded Turing machine $\mathcal{M}$ that generates a spiking neural network $\mathcal{S}_i$ for every input $i$, such that $\mathcal{S}_i$ decides whether $i \in L$ using resources at most $R_S$. Note that in this definition the workload is {\em shared} between the Turing machine $\mathcal{M}$ and the network $\mathcal{S}_i$, and that the definition naturally allows for trading off generality of the network (accepting different inputs by the same network) and generality of the machine (generating different networks for each distinct input), with the traditional Turing machine and family of Boolean circuits being special cases of this trade-off. We can informally see the Turing machine $\mathcal{M}$ as a sort of {\em pre-processing} computation generating the spiking neural network $\mathcal{S}_i$ and then deferring the actual decision to accept or reject the input to this network. We will use the notation $\mathcal{S}(R_T,R_S)$ to refer to the class of decision problems that can be decided in this way. 

There is typically a trade-off between generality and efficiency of a network. Figure \ref{search_snn} provides a simple comparison between three implementations of the {\sc Array Search}-problem: given an array $A$ of integers and a number $i$, does $A$ contain $i$? Note that in the rightmost example a `circuit approach' is emulated. There is no straightforward way to simulate the entire computation for arrays of arbitrary size in the network other than simulating the behaviour of the machine and its input as per the proof in \cite{Maass96}.

\begin{figure*}[t]
    \centering
    \begin{subfigure}[b]{0.32\textwidth}
        \includegraphics[width=\textwidth]{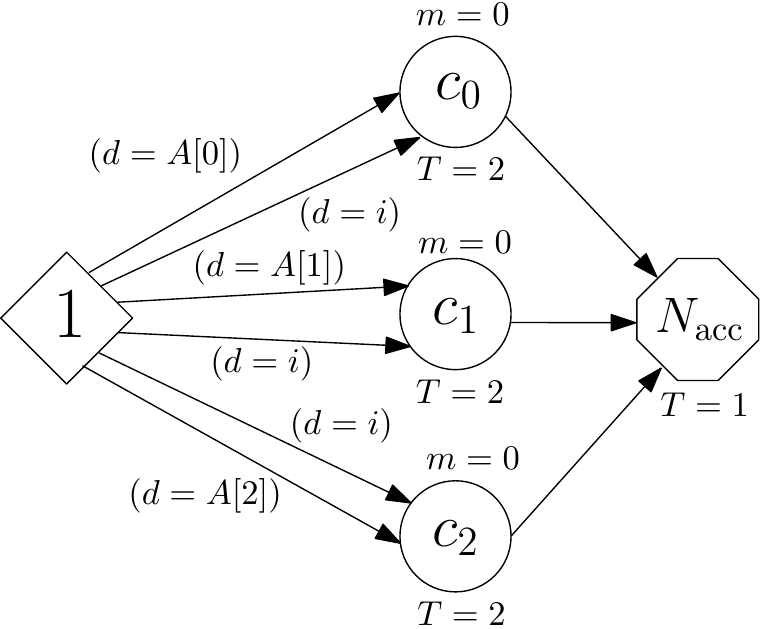}
        \caption{All computation in the network}
        \label{fig:search_allnw}
    \end{subfigure}
    \hfill
    \begin{subfigure}[b]{0.32\textwidth}
        \includegraphics[width=\textwidth]{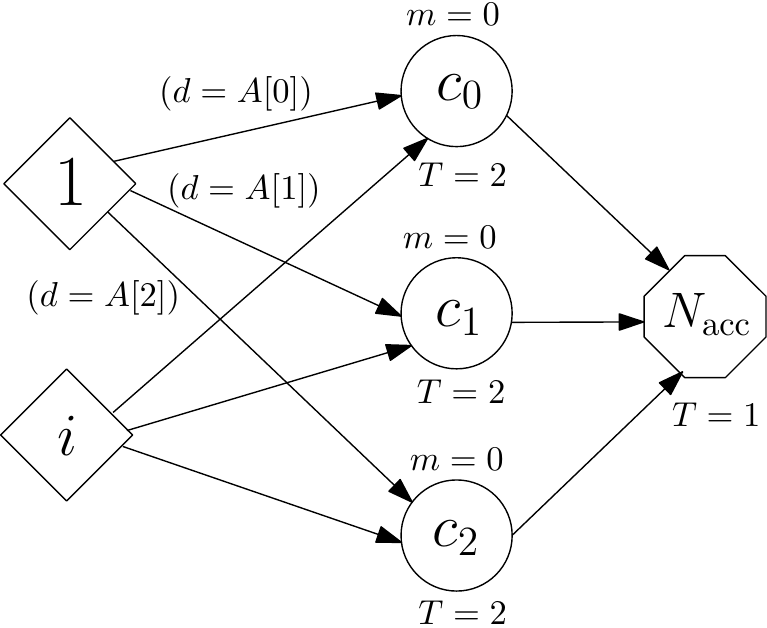}
        \caption{The value $i$ offered to the network as input}
        \label{fig:search_arraynw}
    \end{subfigure}
    \hfill
    \begin{subfigure}[b]{0.32\textwidth}
        \includegraphics[width=\textwidth]{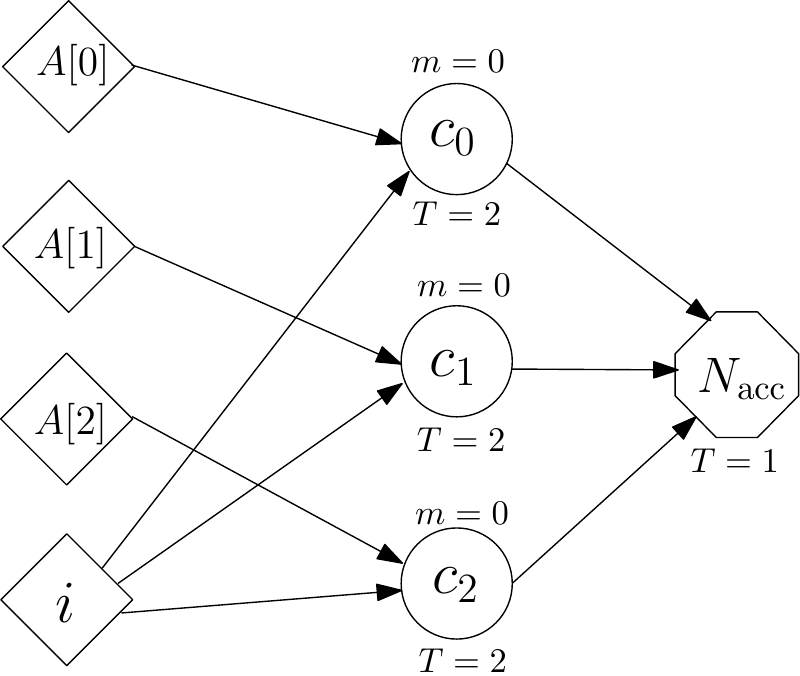}
        \caption{Both the value $i$ and the array $A$ offered to the network as input}
        \label{fig:search_compnw}
    \end{subfigure}
    \caption{Three spiking neural networks designed to decide whether an array $A$ of natural numbers contains $i$. Note that in network \ref{fig:search_allnw} both $A$ and $i$ as well as the parallel comparison are encoded in the network; in network \ref{fig:search_arraynw} the search value $i$ is offered as input (using a spike train consisting of a single spike with delay $i$), and in network \ref{fig:search_compnw} both the search value and the integers in the array are offered as input to the network, while the size of the array is fixed. The number of spikes used in the computation is respectively $\leq |n| + 2$, $\leq |n| + 3$, and $\leq 2|n| + 2$; the generality of the network increases but this comes at the prize of the increasing number of spikes in the computation.}
    \label{search_snn}
\end{figure*}

In addition to the `pre-processing' model we can also allow an {\em iterative} interaction between $\mathcal{M}$ and an oracle capable of deciding whether a spiking neural network $\mathcal{S}$ accepts, such that the computation carried out by $\mathcal{M}$ is interleaved with oracle calls whose results can be acted on accordingly. Before we can properly define this {\em interactive} model of neuromorphic computation, we will first discuss the class $\mathcal{S}(R_T,R_S)$ in further detail. In Section \ref{Complexity_classes} we will cover the formal aspects involved in these definitions; we start by considering the resources that we wish to allocate to these machines.

\section{Resources}
\label{Resources}

We denote the resource constraints of the Turing machine with the tuple $R_T = (\mathsf{TIME}, \mathsf{SPACE})$. We allow the decision of the network to take resources $R_S$; this can be further specified to be a tuple $R_S = (\mathsf{TIME}, \mathsf{SPACE}, \mathsf{ENERGY})$, referring to the number of time steps $\mathcal{S}_i$ may use, the total network size $|\mathcal{S}_i|$, and total number of spikes that $\mathcal{S}_i$ is allowed to use, all as a function of the size of the input $i$. Note that in practice $\mathsf{ENERGY} \leq \mathsf{TIME} \times \mathsf{SPACE}$ since any neuron can fire at most once per time step. Furthermore, note that similarly $\mathsf{SPACE}$ is upper bounded by $R_T$, as for example we cannot in polynomial time construct a network with an exponential number of neurons. We assume in the remainder that the constraints can be described by their asymptotic behavior, and in particular that they are closed under scalar and additive operations; we will describe $R_T$ and $R_S$ as being {\em well-behaved} if they adhere to this assumption. (To clarify, here we restrict ourselves to considering only deterministic resources for both $R_T$ and $R_S$, just as we consider only deterministic membrane potential functions.) 

Observe that we really need the pre-processing to be part of the definition of the model for neuromorphic computation to meaningfully define resource-bounded computations, as we are allowed in principle to define a unique network per instance $i$. Otherwise, the mapping between $i$ and $\mathcal{S}_i$ could be the trivial and uninformative mapping: 

$$i \rightarrow \mathcal{S}_i:  \left\{
\begin{array}{ll}
(N_{\{\mathrm{acc}\}},\varnothing) & \mbox{if $i \in L$;} \\
(N_{\{\mathrm{rej}\}},\varnothing) & \mbox{otherwise.} \\
\end{array}
\right.$$

\subsection{Clock and meter}

We will assume that all Turing machines $\mathcal{M}$ have access to a {\em clock} and a {\em ruler} and enter their rejection state immediately when these bound are violated \cite{Hartmanis65}. In a similar vein, it is possible to build into a spiking neural network $\mathcal{S}$ both a {\em meter} to monitor energy usage as well as a {\em timer} which counts down the allotted time steps, though they will not be part of our baseline assumption. Given an upper bound $e$ on the number of spikes, we can construct an energy counter neuron $E = (e,e,1)$ with synapses $S_{k,E} = (1,1)$ for all $k\in N$, and $S_{E,N_{\mathrm{acc}}} = (1, -\sum_j |w_{jN_{\mathrm{acc}}}|)$, $S_{E,N_{\mathrm{rej}}} = (1, T_{N_{\mathrm{acc}}} + \sum_j |w_{jN_{\mathrm{acc}}}|)$ where applicable. This ensures that if at some time step the permitted number of spikes has been reached without accepting or rejecting (which itself involves a spike from the corresponding neuron), from the next time step on the energy counter will inhibit the acceptance neuron and excite the rejection neuron if present.
Similarly, given an upper bound $t$ on the number of time steps, we can include a programmed timer neuron $T = (1,0,1)$ which fires once at the first time step, along with synapses $S_{T,N_{\mathrm{acc}}} = (t+1, -\sum_j |w_{jN_{\mathrm{acc}}}|)$, and $S_{T,N_{\mathrm{rej}}} = (t+1,  T_{N_{\mathrm{acc}}} + \sum_j |w_{jN_{\mathrm{acc}}}|)$ where applicable (Figure \ref{clock_meter}). Observe that these constructions add only two neurons, a proportionate number of synapses, and (in the presence of a rejection neuron) only a few additional spikes expended, hence the network size and in particular its construction time remain the same asymptotically. 

\begin{figure}[h]
\centering
\includegraphics[width=8cm]{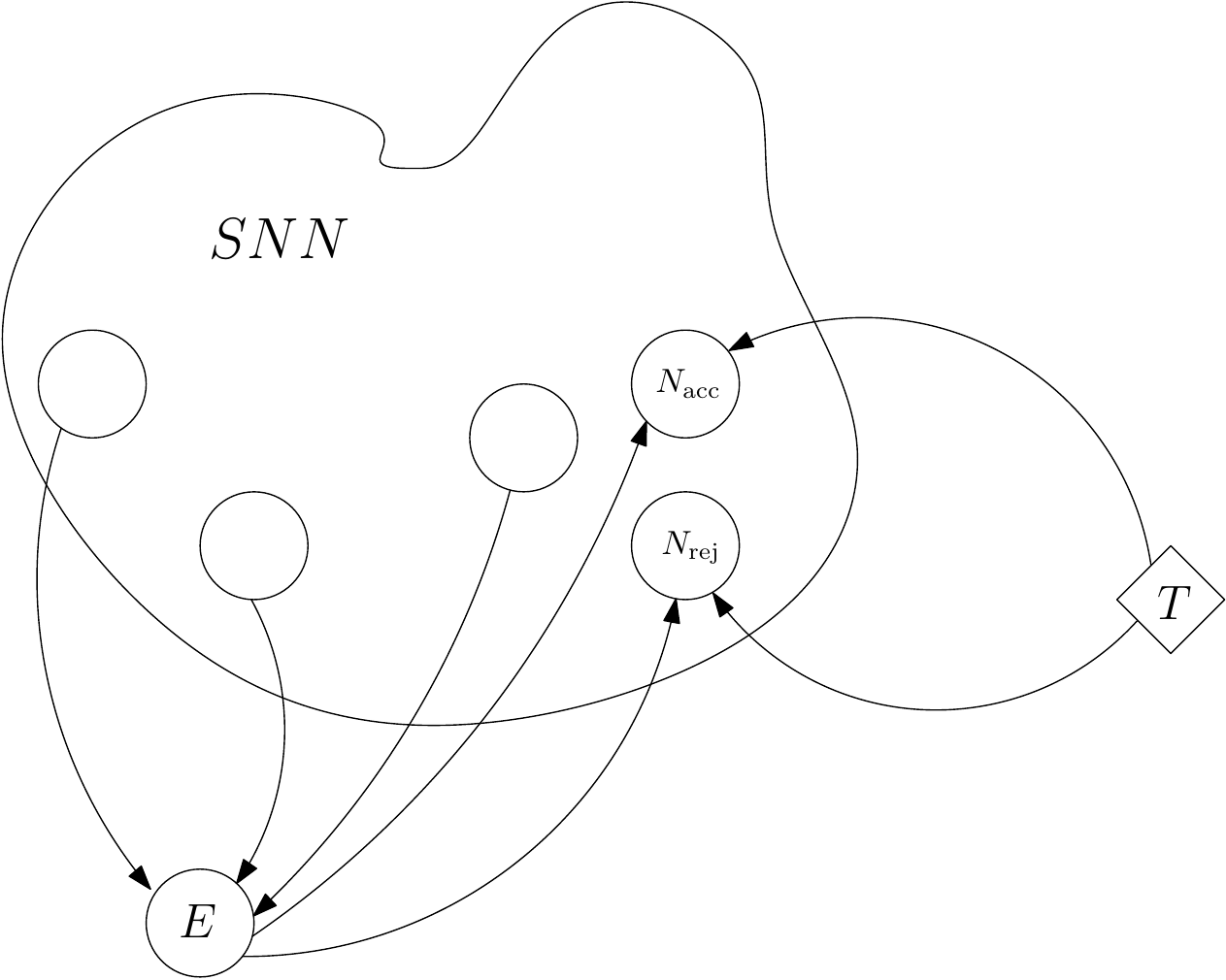}
\caption{Adding a timer and a meter to an arbitrary spiking neural network}
\label{clock_meter}
\rule{\columnwidth}{0.3mm}
\vspace{-1mm}
\end{figure}

\section{Structural complexity}
\label{Complexity_classes}

Now that we have specified what we mean by the resources $R_T$ and $R_S$, it is time to take a closer look at the class $\mathcal{S}(R_T,R_S)$, starting with some initial observations. To begin with, it makes little sense to allow the pre-processing to operate with at least as much resources as the spiking neural network, since otherwise the execution of the spiking neural network can be simulately classically; this remark is illustrated in Theorem~\ref{Plow} below. For this reason we typically choose $R_T$ to be only polynomial time and polynomial or even logarithmic space, corresponding to the classes $\mathsf{P}$ and $\mathsf{L}$ respectively. When the constraints $R_T$ are such that $\mathcal{M}(R_T)$ characterizes familiar complexity classes we will use the common notation for that class from here on; as an abuse of notation we will also use this notation as a shorthand for the resources $R_T$ themselves. 

\begin{theorem}\label{Plow}
${\mathcal{S}(\mathsf{P},R_S)} = \mathsf{P}$ whenever $R_S$ involves at most polynomial time constraints.
\end{theorem}

\begin{proof}
As $\mathsf{P} \subseteq \mathcal{S}(\mathsf{P},R_S)$ is obvious, we focus on proving the inclusion in the other direction. The crucial observation is that for a Turing machine with polynomial time constraints it is impossible to construct a larger than polynomial network, rendering the space constraints actually imposed moot. Recalling our earlier observation that the energy consumption of a spiking neural network is upper bounded in terms of (the product of) its size and time constraints, this implies that the spiking neural network constructed is effectively polynomially bounded (or worse) on all resources. Now it suffices to show that a deterministic Turing machine can simulate in polynomial time the execution of a spiking neural network of polynomial size for at most polynomial time. This can be done by explicitly iterating over the neurons for every time step, determining whether they fire and scheduling the transmission of this spike along the outgoing synapses, until the network terminates or the time bounds are reached. By thus absorbing the decision procedure carried out by the network into the classical polynomial-time computation carried out by the machine we arrive at the stated inclusion.
\end{proof}

This theorem serves as a reminder that spiking neural networks are no magical devices: while there is a potential efficiency gain, mostly in terms of energy usage relative to computations on traditional hardware (only), neuromorphic computations with at most polynomial time constraints cannot achieve more than their classical counterparts. It remains to be determined to what extent the classes $\mathcal{S}(R_T,R_S)$ exhibit any hierarchical behavior based on the constraints $R_S$: in particular, it is still unclear whether there is an {\em energy hierarchy} analogous to the classical time hierarchy. We can however note that for well-defined resource contraints the classes $\mathcal{S}(R_T,R_S)$ are closed under operations such as intersection and complement, since spiking neural networks themselves are, so that decision procedures can be adjusted or combined at the network level.

Observe that using different resource constraints $R_T$ and $R_S$ we can define a lattice of complexity classes $\mathcal{S}(R_T,R_S)$, including such degenerate cases as $\mathcal{S}((\mathcal{O}(1),\mathcal{O}(1)),R_S)$ where the constructed network is only finitely dependent on the actual input (and thus can be constructed in constant time), and $\mathcal{S}(R_T,(\mathcal{O}(1),\mathcal{O}(1),\mathcal{O}(1))) = \mathcal{M}(R_T)$. It is therefore natural to consider the notions of reduction and hardness in this context, which is what we will do next.

\subsection{Completeness for $\mathcal{S}(R_T,R_S)$}

In order to arrive at a canonical complete problem for the class $\mathcal{S}(R_T,R_S)$, it makes sense to consider the analogy with other models of computation, where one asks whether the given procedure (be it machine, circuit or otherwise) accepts the provided input. Since even for the class $\mathcal{S}(R_T,R_S)$ it is not a spiking neural network but a Turing machine which controls how the input is handled, the resulting candidate for a complete problem for this class will involve the latter and not the former. This means that to distinguish this problem from its classical equivalent we must include the promise that the Turing machine is indeed of the kind associated with the class $\mathcal{S}(R_T,R_S)$, in that it generates an $R_S$-bounded spiking neural network using resources $R_T$\footnote{This construction is similar to the one required for the class $\mathsf{BPP}$ associated with probabilistic Turing machines.}. In other words, we claim that the following problem is complete under polynomial-time reductions for the promise version of the class $\mathcal{S}(R_T,R_S)$. 

\pproblem
{\SNNHalting}
{Turing machine $\mathcal{M}$ along with input string $i$}
{$\mathcal{M}$ is an ${\mathcal{S}(R_T,R_S)}$-machine}
{Does $\mathcal{M}$ accept $i$}

\begin{theorem} 
\SNNHalting\ is complete under polynomial-time reductions for the promise version of $\mathcal{S}(R_T,R_S)$.
\end{theorem}

\begin{proof}
Membership of this problem is established as follows: with a universal $\mathcal{S}(R_T,R_S)$ machine one can take the machine $\mathcal{M}$ and simulate it on the input $i$. If the machine $\mathcal{M}$ is indeed an $\mathcal{S}(R_T,R_S)$ machine as per the promise, then this simulation will succeed within the permitted resource bounds and we can simply return the answer given by $\mathcal{M}$. In case the promise fails to hold, we only need to ensure that the (unsuccessful) simulation does not exceeds the resource bounds, since it is otherwise irrelevant which response is ultimately given.
For the hardness of this problem, we observe that every problem in $\mathcal{S}(R_T,R_S)$ is by definition solvable by an $\mathcal{S}(R_T,R_S)$-machine, hence the straightforward reduction from any such problem to \SNNHalting\ consists of taking the input $i$ and passing it along to \SNNHalting\ accompanied by a particular $\mathcal{S}(R_T,R_S)$-machine which decides the problem.
\end{proof}

However, for particular assignments of $R_T$ we can actually replace the Turing machine by a spiking neural network and still end up with a complete (promise) problem. We will illustrate this construction for $R_T$ being linear time (and space); the same result also holds for $R_T$ corresponding to $\mathsf{P}$ and $\mathsf{L}$ under polynomial-time and logspace reductions respectively. 

\pproblem
{\SNNAccept}
{Network $\mathcal{S}$ along with input string $i$}
{$\mathcal{S}$ terminates within resource bounds $R_S$ expressed as a function of $|i|$}
{Does $\mathcal{S}$ accept}

\begin{theorem} 
\SNNAccept\ is complete under linear-time reductions for the promise version of\\$\mathcal{S}((\mathcal{O}(n),\mathcal{O}(n)),R_S)$.
\end{theorem}

\begin{proof}
Membership follows from the observation that a Turing machine can in linear time discard the input string $|i|$, such that what it is left with is a network promised to be $R_S$-constrained that accepts precisely when $\mathcal{S}$ does as it is $\mathcal{S}$ itself.
To prove hardness we reduce $\mathcal{S}((\mathcal{O}(n),\mathcal{O}(n)),R_S)$-{\sc Halting} to \SNNAccept. Let $(\mathcal{M},i)$ be an instance of the former. By simulating the application of $\mathcal{M}$ on $i$ and replacing it with the resulting network $\mathcal{S}_i$ (which by the promise can be done in linear time), we obtain an instance $(\mathcal{S}_i,i)$ of $\mathcal{S}((\mathcal{O}(n),\mathcal{O}(n)),R_S)$-{\sc Network Halting} where the promise for $\mathcal{S}_i$ is inherited from that for $\mathcal{M}$ and the decision of $\mathcal{S}_i$ is that of $\mathcal{M}$ on $i$ by definition.
\end{proof}

This completeness result shows that for those choices of $R_T$ that we were likely to consider anyways (cf.~the remark at the beginning of this section) we are justified in taking spiking neural networks as computationally primitive in a sense relevant for our treatment. In particular, this allows us to round off our discussion by exploring the interactive model of neuromorphic computation.

\subsection{Interactive computation}

We will formalize the interactive model of neuromorphic computation in terms of Turing machines equipped with an oracle for the relevant class of spiking neural networks. This involves augmenting a deterministic Turing machine with a {\em query tape}, an {\em oracle-query} state, and an {\em oracle-result} state. We can then select the problem $\mathcal{S}(R_T,R_S)$-{\sc Network Halting} for our choice of $R_T$ and $R_S$ to serve as an oracle to our machine. Now when a machine with such an oracle enters the oracle-query state with $(\mathcal{S},i)$ on its query tape it proceeds to the oracle-result state, at which point it will replace the contents with $1$ if $\mathcal{S}$ accepts and with $0$ if $\mathcal{S}$ rejects (given that the promise holds; the outcome otherwise returned is unspecified).
With a slight abuse of notation, we can thus define $\mathcal{M}(R_{T'})^{\mathcal{S}(R_T,R_S)}$ to be the class of decision problems that can be solved by a Turing machine with resource constraints $R_{T'}$ equipped with an oracle for $\mathcal{S}(R_T,R_S)$-{\sc Network Halting}. It follows immediately that $\mathcal{M}(R_{T'})^{\mathcal{S}(R_T,R_S)}$ is a superclass of $\mathcal{S}(R_{T'},R_S)$, though again the exact relations between these two kinds of classes and between these neuromorphic complexity classes and the classical complexity classes remain to be determined. In closing we can however offer an example of a potential use for the interactive model of neuromorphic computation.

\begin{example}
Suppose we are interested in the behavior of $\mathsf{P}$-complete problems on neuromorphic oracle Turing machines. Given that such problems are assumed to be inherently serial and cannot be computed with only a logarithmic amount of working memory, one might suggest to look at a suitable trade-off between computations on a regular machine and on a neuromorphic device. One way of doing this would be to constrain the working memory for the Turing machine to be logarithmic in the input size, so that $\mathcal{M}(R_T)$ characterizes the complexity class $\mathsf{L}$. Then if all resources $R_S$ are linear in the size of the input, we obtain the complexity class $\mathsf{L}^{\mathcal{S}(\mathcal{O}(n),\mathcal{O}(n),\mathcal{O}(n))}$. In a related paper we will show that indeed the $\mathsf{P}$-complete \NwFlow\ problem resides in this class \cite{Ali19}.
\end{example}

\section{Conclusion}
\label{Conclusion}

In this paper we proposed a machine model to assess the potential of neuromorphic architectures with energy as a vital resource in addition to time and space. We introduced a hierarchy of computational complexity classes relative to these resources and provided some first structural results and canonical complete problems for these classes.

We already hinted at some future structural complexity work, most urgently an energy-analogue for the time complexity hierarchy and a notion of amortization of resources. The latter is crucial when considering {\em local changes} to the network, such as adapting the weights when learning, or when using a network with a set of spike trains rather than recreating everything from scratch.

In addition, providing concrete hardness proofs, as well as populating classes using neuromorphic algorithms, should be high on the agenda for the neuromorphic research community.

\bibliographystyle{ACM-Reference-Format}
\bibliography{SNN_complexity}

\end{document}